\PassOptionsToPackage{warn}{textcomp}
\documentclass[conference]{IEEEtran}
\usepackage{times}
\usepackage{graphicx} % more modern
\usepackage{subfig}

\usepackage{algorithm}
\usepackage{algorithmic}

\usepackage{amsmath,amssymb,amsthm}
\newtheorem{theorem}{Theorem}
\newtheorem{lemma}{Lemma}

\usepackage{hyperref}
\usepackage{xcolor}
\let\labelindent\relax % added for NDSS
\usepackage{enumitem}
\usepackage{balance}

\newcommand\TODO[1]{\textcolor{red}{TODO: {#1}}}

\begin{document} 
\title{Differentially-Private ``Draw and Discard'' Machine Learning}

\author{\IEEEauthorblockN{Vasyl Pihur\IEEEauthorrefmark{1},
Aleksandra Korolova\IEEEauthorrefmark{2},
Frederick Liu\IEEEauthorrefmark{3}, 
Subhash Sankuratripati\IEEEauthorrefmark{3}, \\
Moti Yung\IEEEauthorrefmark{4}, 
Dachuan Huang\IEEEauthorrefmark{3}, 
Ruogu Zeng\IEEEauthorrefmark{3}} 
\IEEEauthorblockA{\IEEEauthorrefmark{1}Corresponding author. Snap Inc.
Email: vpihur@snapchat.com}
\IEEEauthorblockA{\IEEEauthorrefmark{2}USC.
Email: korolova@usc.edu}
\IEEEauthorblockA{\IEEEauthorrefmark{3}Snap Inc., Emails: \{fliu, subhash, dachuan.huang, ruogu.zeng\}@snap.com}
\IEEEauthorblockA{\IEEEauthorrefmark{4}Snap Inc., Email: moti@cs.columbia.edu}}

\IEEEoverridecommandlockouts
\makeatletter\def\@IEEEpubidpullup{9\baselineskip}\makeatother
\IEEEpubid{\parbox{\columnwidth}{Anonymous submission to IEEE Symposium on Security and Privacy 2019.
}
\hspace{\columnsep}\makebox[\columnwidth]{}}

\maketitle
\begin{abstract}
In this work, we propose a novel framework for privacy-preserving client-distributed machine learning. It is motivated by the desire to achieve differential privacy guarantees in the \emph{local} model of privacy in a way that satisfies all systems constraints using \emph{asynchronous} client-server communication and provides attractive model learning properties. We call it ``Draw and Discard'' because it relies on random sampling of models for load distribution (scalability), which also provides additional server-side privacy protections and improved model quality through averaging. We present the mechanics of client and server components of ``Draw and Discard" and demonstrate how the framework can be applied to learning Generalized Linear models. We then analyze the privacy guarantees provided by our approach against several types of adversaries and showcase experimental results that provide evidence for the framework's viability in practical deployments. We believe our framework is the first deployed distributed machine learning approach that operates in the local privacy model.
\end{abstract}

%\keywords{Machine learning, differential privacy, SGD, distributed learning}

\section{Introduction}
%\section{Introduction}
In this work, we propose a Machine Learning (ML) framework, unique in many ways, that touches on several different aspects of \emph{practical} deployment of locally differentially private ML, all of which are equally important. These aspects include feasibility, scalability, efficiency, spam protection, ease of implementation, and, of course, privacy. Ideally, they all must interplay together in a manner that enhances each other. From that perspective, this work is as much a systems one, as it is both privacy and machine learning focused.

Machine learning made our mobile devices ``smart''. Applications span a wide range of seemingly indispensable features, such as personalized app recommendations, next-word suggestions, feed ranking, face and fingerprint recognition and many others. The downside is that they often come at the expense of privacy of the users who share their personal data with parties providing these services. However, as we demonstrate in this work, this does not necessarily need to be the case.

Historically, ML was developed with a server-centric view of first collecting data in a central place and then training models based on them. Logistic regression~\cite{cox} and neural networks~\cite{neural}, introduced
over half a century ago, follow a now-familiar paradigm of reading training data from the local disk and
adjusting model weights until certain convergence criteria are met. With the widespread use of mobile devices
capable of generating massive amounts of personal information backed up by the convenience of cloud data storage 
and infrastructure, the community adopted the \emph{server-centric} world-view of ML simply because it was convenient to do so.
As the training data grew in size and could no longer fit on a single machine or even several machines, we ended up collecting data
from millions of devices on one network and sending it ``sharded'' for training to another network of thousands of machines. 
In the past, this duality of responsibilities could be justified by large disparities in hardware capabilities between the two networks,
but this line is blurrier at the present time.

Sharing personal data that contributes to a global ML model and benefits everyone on the network---in many cases, the data collector the most--
can be viewed as undesirable by many privacy sensitive users, due to distrust in the data collector or risks of subpoenas, data breaches and internal threats~\cite{pew2015privacy, pew2016privacy}. Following the deployment of RAPPOR~\cite{Rappor}, there has been an increased interest in finding ways for users to contribute data to improve services they receive, but to do so in a provably private manner, even with respect to the data collector itself~\cite{eff}. This desire is often expressed by companies~\cite{wired, applekeynote2016}, presumably in part to minimize risks and exposures.

To address the privacy-utility trade-off in improving products while preserving privacy of user data even from the data collector itself, we propose a novel client-centric distributed ``Draw and Discard" machine learning framework (DDML). It provides differential privacy guarantees in the \emph{local} model of privacy in a way that satisfies all systems constraints using \emph{asynchronous} client-server communication. We call it ``Draw and Discard'' because it relies on randomly sampling and discarding of models. Specifically, DDML maintains $k$ versions (or instances) of the machine learning model on a server, from which an instance is randomly selected to be updated by a client, and, subsequently, the updated instance randomly replaces one of the $k$ instances on a server. The update is made in a differentially private manner with users' private data stored locally on their mobile devices. 

We focus our analyses and experiments with DDML on the Generalized Linear Models (GLM)~ \cite{glm}, which include regular linear and logistic regressions. GLMs provide widely-deployed and effective solutions for many applications of non-media rich data, such as event prediction and ranking. The convex nature of GLMs makes them perfect candidates to explore client-side machine learning without having to worry about convergence issues of the more complex ML models. Extension of DDML to Neural Networks and other models optimized through iterative gradient updates is relatively straightforward.

We demonstrate through modeling, analyses, experiments and practical deployment that DDML provides attractive privacy, model learning and systems properties (feasibility, scalability, efficiency and spam protection). Specifically, it offers
\begin{enumerate}[noitemsep,topsep=0pt]
\item \emph{Local differential privacy}: Through carefully calibrated noise in the model update step, the DDML design ensures local differential privacy~\cite{dwork2006calibrating}. 
\item \emph{Privacy amplification against other adversaries}: Furthermore, in DDML the full model update is performed on a client and only the updated model rather than raw gradients are sent to the server, which strengthens the privacy guarantees provided by DDML against weaker but realistic adversaries than the strongest possible adversary operating in the local model.
\item \emph{Efficient model training}: Due to the variance stabilizing property of DDML, its final model averaging and relatively frequent model updating, DDML has superior finite sample performance relative to alternative update strategies.
\item \emph{Asynchronous training}: continuous, lock-free and scalable training without pausing the process for averaging and updating on a server side.
\item \emph{Spam protection}: having $k$ different instances of the same model allows to assess whether any incoming update is fraudulent or not without knowledge of users' private data.
\item \emph{Limited server-side footprint}: store relatively little data on a server at any given time, since $k$ is usually small.
\end{enumerate}

These properties will become clearer as we define them more precisely in the following sections. 

We are not aware of any other, currently deployed distributed ML approaches, that operate in the local privacy model and use an asynchronous communication protocol. Federated Learning~\cite{McMahan2017}, which is adopted by Google and is perhaps the closest alternative, relies on server-side gradient batching and averaging.

DDML offers two major contributions. Fist, it performs direct, \emph{noisy} updates of model weights on clients, as opposed to sending \emph{raw and exact} gradients back to the server. This change offers \emph{local} differential privacy guarantees, and more importantly, requires an attacker to know the pre-image of the model (a model sent for an update to the client) that was updated to make any inference about private user data. Separation of the two critical pieces of knowledge, pre- and post-update models necessary to make any inference, \emph{in time}, especially in a high-throughput environment with $k$ instances being continuously updated, poses significant practical challenges for an adversary observing a stream of updates on the server side. We discuss this in detail in Section~\ref{sec:privacy}.

Second, its radically different server-side architecture with the ``Draw and Discard'' update strategy provides a natural way of deploying such a service in the cloud. Sharding and replication in the cloud~\cite{replication} is necessary to avoid updating entities or values too often, beyond their specified or recommended number of times per second.  Making too frequent writes to a single piece of data, be it in datastore, memcache~\cite{memcache2} or elsewhere, results in what is known as ``hotspotting''~\cite{hotspotting}. If you find yourself in a situation where several keys become ``hot'' (updated beyond their capacity), your service becomes unstable and clearly cannot accept even more traffic (unscalable). Even worse, it often negatively effects services co-hosted with yours in the cloud that you may not even own. 

Beyond these two major considerations, DDML offers a completely lock-free asynchronous, and thus, more efficient, communication between the server and clients, which is an absolute must if one is developing in a massively distributed environment~\cite{distributed}, as well as a straightforward distributed way to prevent model spamming by malicious actors, without sacrificing user privacy.

We have implemented DDML at a large tech company and successfully trained many ML models. Our applications focus on ranking items, from a few dozen to several thousands, as well as security oriented services, such as predicting how likely it is that a URL one receives from someone is phishy. Our largest models contain $\approx 50,000$ weights in size and we find $k=20$ to be the right trade-off between efficiency and scale to avoid the ``hotspotting'' issue. Currently, at peak times for several different applications, we receive approximately 200 model updates per second.

The paper is organized as follows: Section~\ref{sec:related} reviews differential privacy and related work. Section~\ref{sec:ddml} presents a detailed overview of our framework and its features, including the variance stabilizing property in Section~\ref{sec:varst}.  Section~\ref{sec:privacy} introduces our modeling of possible adversaries and discusses DDML's privacy properties with respect to them. In Section~\ref{sec:experiments}, we present experimental evaluations of DDML's performance as compared to alternatives and in Section~\ref{sec:applications} we describe the performance of a real-world application deployed using DDML. We conclude with a discussion of limitations, alternatives and avenues for future work in Section~\ref{sec:discussion}.

\section{Related Work}\label{sec:related}
%\section{Related Work}
Differential privacy (DP)~\cite{dwork2006calibrating} has become the de-facto standard for privacy-preserving data analysis~\cite{dwork2014algorithmic, dwork2011acm, Godel}.

A randomized algorithm $A$ is $(\epsilon, \delta)$ differentially private if for all databases $D$ and $D'$ differing in one user's data, the following inequality is satisfied for all possible sets of outputs $Y \subseteq Range(A)$:
\[Pr[A(D) \in Y] \leq \exp(\epsilon) Pr[A(D') \in Y] + \delta\]

The parameter $\epsilon$ is called the privacy loss or privacy budget of the algorithm~\cite{nissim2017differential}, and measures the increase in risk due to choosing to participate in the DP data collection. The variant of DP when $\delta=0$ is the strongest possible differential privacy variant called \textit{pure} differential privacy; whereas $\delta>0$ allows differential privacy to fail with small probability and is called \textit{approximate} differential privacy.

%\TODO{DP definition: both E, 0 and E, delta. Call the delta one approximate differential privacy using Dwork 2006}

%\TODO{Discuss why the local model of privacy is more desirable but harder to achieve}.

%\TODO{Privacy loss as notation for epsilon, constants matter (cite latest PATE), also the fact that privacy risk goes up exponentially with epsilon. Expectation over the random coin tosses of the algorithm}

\textbf{ML in the Trusted-Curator Model:}
Most prior work for differentially private machine learning assumes a trusted-curator model, where the data is first collected by the company and only then a privacy-preserving computation is run on it \cite{Abadi2016, song2013stochastic, papernot2016, chaudhuri2009privacy}. The trusted-curator model is less than ideal from the user privacy perspective, as it does not provide privacy from the company collecting the data, and, in particular, leaves user data fully vulnerable to security breaches, subpoenas and malicious employees. Furthermore, even in the case of the trusted curator model, differentially private deep learning that achieves good utility with reasonable privacy parameters has been an elusive goal~\cite{shokri2015privacy, mcsherry2017, Abadi2016}. For example, the work of~\cite{Abadi2016} performs well on NIST data but struggles utility-wise on CIFAR for reasonable privacy parameters. % the work of~\cite{shokri2015privacy} results in privacy loss that grows with the number of parameters of the model, which is considered unacceptable by the theoretical differential privacy community~\cite{mcsherry2017}.

% The work closest to ours is Federated Learning ~\cite{McMahan2017, mcmahanlanguage}. We will further elaborate on the distinctions and properties of DDML in Sections~\ref{sec:ddml-properties} and Section~\ref{sec:privacy}.

\textbf{ML in the Local Model:}
The pioneering work of RAPPOR~\cite{Rappor} for industry deployment, has been followed by several recent efforts to deploy DP in the local model, i.e., guarantee DP to the user before the data reaches the collector. Privacy in the local model is more desirable from the user's perspective~\cite{wired, pew2015privacy, eff, applekeynote2016}, as in that case the user does not have to trust the data collector or the data being subject to internal or external threats to the data collector.

Since the focus on differentially private computations in the local model is recent, most, if not all, efforts to date have been limited to learning aggregate statistics, rather than trainng more complex machine learning models \cite{fanti2016building, Apple, bassily2015local, bassily2017practical, bun2017heavy}. There are also numerous results on the so-called sample complexity for the local model, showing that the number of data points needed to achieve comparable accuracy is significantly higher in the local model than in the trusted curator model~\cite{kairouz2014extremal}. 

DDML can be considered an extension of the existing literature on locally private learning. In particular, it supplements private histogram collection of RAPPOR \cite{Rappor} and learning simple associations~\cite{fanti2016building} by allowing estimation of arbitrary conditional expectations. While RAPPOR allows to estimate marginal and joint distributions of categorical variables, DDML provides a principled framework for estimating conditional distributions in a privacy-preserving manner. For example, one can estimate the average value of $Y$ given $p$ features $X_1, \ldots, X_p$ by fitting a regular linear model described by
\begin{equation*}
E(Y) = b_0 + b_1 * X_1 + \ldots + b_p * X_p.
\end{equation*}

\section{Draw and Discard Machine Learning}\label{sec:ddml}
%\section{Draw and Discard Machine Learning}
In this section, we present our ``Draw and Discard'' machine learning framework characterized by its two major components: client-side noise addition and ``Draw and Discard'' server architecture. Together, these contribute to strong differential privacy guarantees for client data while supporting efficient, in terms of model training, client-server data consumption. 

At the heart of DDML is the server-side idea of maintaining and randomly updating one of the $k$ model instances. This architecture presents a number of interesting properties and contributes to many aspects of the framework's scalability, privacy, and spam and abuse protections.

DDML is not a new ML framework \textit{per se}. It is model-agnostic and, in principle, works with any supervised ML model, though details vary for the client-side update and noise addition. The scope of this work is limited to Generalized Linear Models (GLMs), and we focus specifically on logistic regression to show an example of an ML model that is very popular in practice. % \TODO{cite}.
We give a brief overview of GLMs and fully describe DDML client and server architectures next.

\subsection{GLMs}
In GLMs \cite{glm}, the outcome or response variable $Y$ is assumed to be generated from a particular distribution in the exponential family that includes normal (regular linear model), binomial (logistic regression) and Poisson (Poisson regression) distributions, among many others. Mathematically, GLMs model the relationship between response $Y$ and features $X_1, \ldots, X_p$ through a link function $g$, whose exact form depends on the assumed distribution:
\begin{equation}
E(Y) = g^{-1}(b_0 + b_1 * X_1 + \ldots + b_p * X_p)
\end{equation}
\vspace{-0.25in}

To train GLM models on clients, we use Stochastic Gradient Descent (SGD) for maximum likelihood estimation, as discussed below.

\subsection{DDML Client-Side Update}
SGD is a widely used iterative procedure for minimizing an objective function
\begin{equation}
Q(B) = \frac{1}{N}\sum_{s=1}^NQ_s(B),
\end{equation}
where $B = \{b_1, \ldots, b_p\}$ is the vector of weights to be estimated and $Q_s$ is a functional component associated with the $s$th observation. Traditional optimization techniques require differentiating $Q(B)$, which, in turn, requires access to all $N$ data points at once. SGD approximates the gradient $\Delta Q(B)$ with $\Delta Q_s(B)$, computed on a small batch of $N_c$ observations available on a single client
\begin{equation}
\Delta Q_{N_c}(B) = \frac{1}{N_c}\sum_{s=1}^{N_c}\Delta Q_s(B).
\end{equation} 
\vspace{-0.15in}

To provide local privacy by adding random Laplace noise, a differentially-private SGD (DP-SGD) update step is performed on a client using the $N_c$ observations stored locally 
\begin{equation}
B^{t+1} = B^t - \gamma \Delta Q_{N_c}(B^t) + L\left(0, \frac{\Delta f}{\epsilon}\right),
\end{equation}
where $\gamma$ is a learning rate and $L\left(0, \frac{\Delta f}{\epsilon}\right)$ denotes the Laplace distribution with mean 0 and scale $\frac{\Delta f}{\epsilon}.$ $\Delta f$ is called sensitivity in the differential privacy literature and $\epsilon$ is the privacy budget \cite{dwork2014algorithmic}. 

For GLMs, assuming all features $X_i$ are normalized to the $[0, 1]$ interval and the average gradients $\frac{1}{N_c}\sum_i^{N_c}(\hat{Y_i} - Y_i)X_i$ are clipped to $[\text{-1}, 1]$ (indicated by $\Vert A \Vert_{[\text{-}1, 1]}$), the differentially-private update step becomes 
\begin{equation*}
B^{t+1} = B^t - \gamma\Big\Vert\frac{1}{N^c}\sum_{i=1}^{N_c}(\hat{Y_i} - Y_i)X_i\Big\Vert_{[\text{-}1, 1]} + L\left(0, \frac{2\gamma}{\epsilon}\right).
\end{equation*}
Here, $\hat{Y}$ is the predicted value of $Y$ given a feature vector $X$ and the model $B^t$. For logistic regression, if all features are normalized to $[0, 1]$, no gradient clipping is necessary.

\begin{algorithm}[tb]
   \caption{DDML Algorithm for GLMs (client side). Parameters: \\$Y$ - response value, $\hat{Y}$ - predicted value.\\ $X$ - feature vector, $B$ - a set of model weights. \\ $\gamma$ - learning rate, $\epsilon$ - privacy budget. \\ $L(\mu, s)$ - Laplace distribution with mean $\mu$ and scale $s$.}
\label{alg}
\begin{algorithmic}
   \STATE Normalize response $Y$ and features $X$ to $[0, 1]$
   \STATE If $N_c > 0$, request model $B^t$ from the server
   \STATE Compute clipped average gradient: $\Vert G\Vert_{[\text{-}1, 1]} = \Vert\frac{1}{N^c}\sum_{i=1}^{N_c}(\hat{Y_i} - Y_i)X_i\Vert_{[\text{-}1, 1]}$
   \STATE Update model: {$B^{t+1} = B^t - \gamma\Vert G\Vert_{[\text{-}1, 1]} + L\left(0, \frac{2\gamma}{\epsilon}\right)$}
   \STATE Return model $B^{t+1}$ to the server
\end{algorithmic}
\label{alg1}
\end{algorithm}

%The sensitivity of this update (for the purpose of differential privacy) is $2\gamma,$ as we either subtract at most $\gamma$ when both the residual and the $X$ are 1 or we add at most $\gamma$ when the residual is $-1$ and the $X$ is 1.\aleks{Maybe move the sensitivity discussion to Section~\ref{sec:privacy-local}}

DDML client-side architecture is shown in Algorithm~\ref{alg1}.

\subsection{DDML Server-Side Draw and Discard}
\begin{algorithm}[tb]
   \caption{DDML Algorithm (server side). Parameters: \\$k$ - the number of models, $B$ - a set of model weights.}
\begin{algorithmic}
   \STATE Initialize $k$ models 
   \FOR{each requested model update $t$}
       \STATE Pick a random model instance $B^t$
       \STATE Send $B^t$ to a client
       \STATE Receive updated $B^{t+1}$ from a client
       \STATE Replace a random instance of the $k$ models with $B^{t+1}$
   \ENDFOR
   \STATE Prediction: average $k$ model instances
\end{algorithmic}
\label{alg2}
\end{algorithm}

While maintaining the $k$ model instances on a server ($k$ versions of the same model with slightly different weights), we randomly ``draw'' one of the $k$ instances, update it on a client and put it back into the queue by ``discarding'' an instance uniformly at random. With probability $\frac{1}{k}$, we will replace the same instance, while with probability $\frac{k-1}{k}$, we will replace a different one.\footnote{This is only approximately correct, since in a high-throughput environment, another client request could have updated the same model in the meantime.}

This seemingly simple scheme has significant practical implications for performance, quality, privacy, and anti-spam, which we discuss in Section~\ref{sec:ddml-properties}.

DDML server-side architecture is shown in Algorithm~\ref{alg2}.

\paragraph{Model Initialization:}
We initialize our $k$ models randomly from a normal distribution with means 
$
b^0_0, \ldots, b^0_p,
$
which are usually taken to be 0 in the absence of better starting values and variance 
$
\sigma^2_k = \frac{k}{2}\sigma^2,
$
where $\sigma^2$ is the variance of the Laplace noise added on a client side.

Because of the variance-stabilizing property (to be discussed in detail in Section~\ref{sec:varst}),  $\sigma^2_k$ will remain the same in expectation even after a large number of updates. It is crucial for our spam detection solution that this initialization happens correctly and the right amount of initial noise is added to calibrate the update step on a client with the variance of the $k$ instances on the server side.

\paragraph{Model Averaging:}
We average weights from all $k$ instances for final predictions. Of course, depending on application, another way for using $k$ versions of the same model could be preferred, such as averaging predicted values from each instance, for example. 

\subsection{Properties and Features of DDML}\label{sec:ddml-properties}
We now describe properties of DDML that distinguish it from existing solutions and make it feasible and scalable for practical deployments.

\begin{figure}[tb]
\centering
\includegraphics[width= 1.0\linewidth]{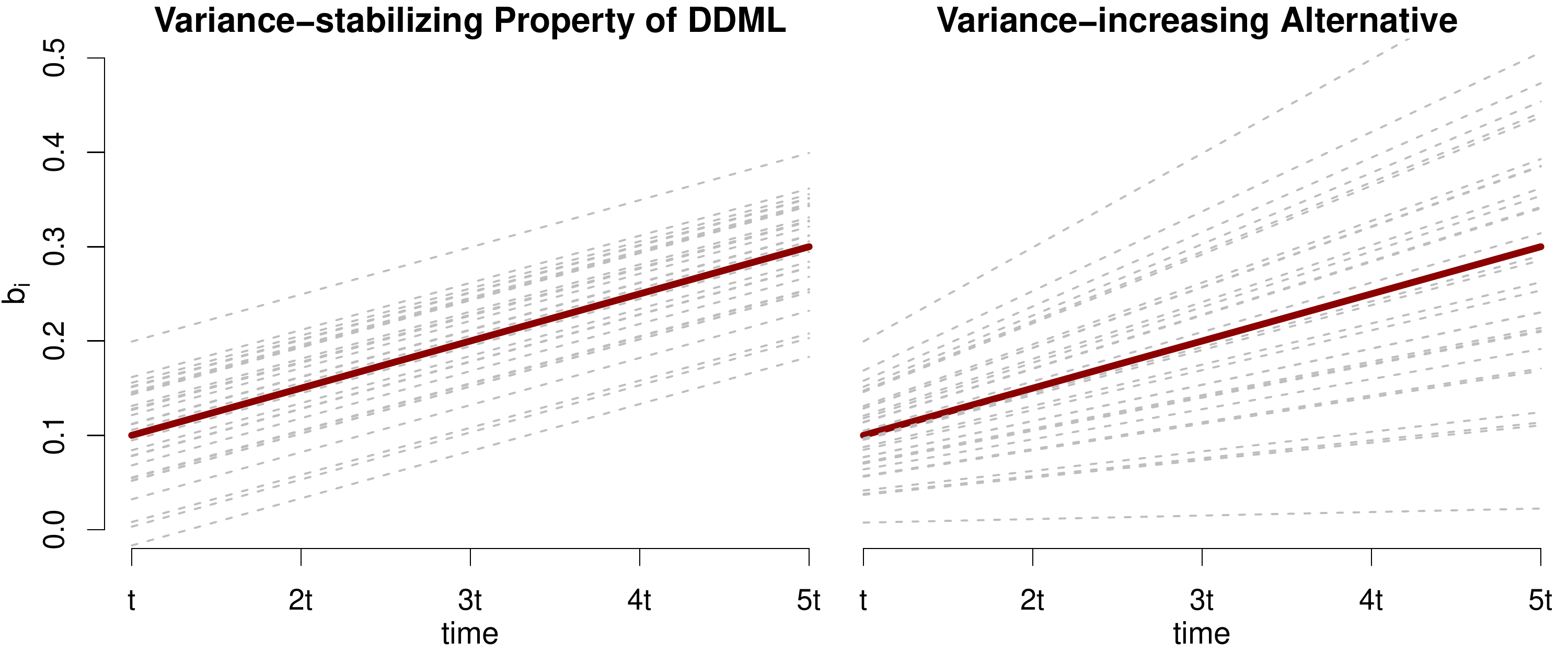}
\caption{Variance-stabilizing property of ``Draw and Discard''. The left panel shows $k$ models having the same intra-model variance around the \emph{average} model (dark red line) over time. This is in contrast to the right panel where they diverge (intra-model variance increases over time).}
\label{variance-stab}
\end{figure}

\subsubsection{Variance-stabilizing Property of DDML}\label{sec:varst}
Having introduced an additional source of variation due to having $k$ model instances, an intra-model variance, it is important to understand its nature and magnitude, especially relative to the variance of the noise added on the client through the Laplace mechanism. It is also of interest to understand how it changes over time.
One of the remarkable properties of the ``Draw and Discard'' algorithm with $k > 1$ is its variance-stabilizing property, shown schematically in the first panel of Figure~\ref{variance-stab}. We prove in Theorem~\ref{thm:vars} that the expected intra-model variance of the $k$ instances is equal to $\frac{k}{2}\sigma^2$ and remains unchanged after an infinite number of updates when adding noise with mean $0$ and variance $\sigma^2$. 

\begin{theorem}\label{thm:vars}
Let there be $k$ models where each weight\\ $B_1, \ldots, B_i, \ldots, B_p$ has a mean $\mu_i$ and variance $\frac{k}{2}\sigma_i^2$ with $i \in [1, p]$. Selecting one model at random, adding noise with mean 0 and variance $\sigma_i^2$ to each $B_i$, and putting the model back with replacement does not change the expected variance of the weights (i.e., they remain distributed with variance $\frac{k}{2}\sigma_i^2$).
\end{theorem}

The intuition behind this theorem is that with probability $\frac{1}{k}$, we replace the same model, which increases the variance of the $k$ instances. This increase, however, is exactly offset by the decrease in variance due to the cases when we replace a different model with probability of $\frac{k-1}{k}$ because original and updated models are highly correlated.

\begin{proof}
We use the Law of Total Variance
$$
Var(Y) = E(Var(Y|X)) + Var(E(Y|X)).
$$

Replacing the same model as drawn occurs with probability $\frac{1}{k}$ and the variance after the update for each $B_i$ is equal to 
\begin{eqnarray*}
Var(B_i| j \to j) &=& \frac{1}{k}\left(\frac{k}{2} + 1\right)\sigma_i^2 + \frac{k-1}{k}\frac{k}{2}\sigma_i^2 \\
        &=& \left(\frac{1}{2} + \frac{1}{k} + \frac{k-1}{2}\right)\sigma_i^2 \\
        &=& \left(\frac{k}{2} + \frac{1}{k}\right)\sigma_i^2.
\end{eqnarray*}

Replacing a different model partitions the space into 2 and $(k-2)$ models which makes $Var(E(Y|X))$ non-zero. The overall mean after the update becomes
\begin{eqnarray*}
\mu_{B_i} &=& \frac{2\mu_1 + (k-2)\mu}{k},  
\end{eqnarray*}
where $\mu_1$ is the mean of the model selected and model replaced and has a distribution with mean $\mu$ and variance $\frac{k}{2}\sigma_i^2$. 

Then $Var(B_i|j \to j')$
\begin{eqnarray*}
         &=& \frac{2}{k}\frac{1}{2}\sigma_i^2 + \frac{k-2}{k} \frac{k}{2}\sigma_i^2 \\
         && + \frac{2E[(\mu_1 - \mu_{B_i})^2] + (k-2)E[(\mu-\mu_{B_i})^2]}{k-1} \\
         &=& \frac{1}{k}\sigma_i^2 + \frac{k-2}{2}\sigma_i^2 \\
         && + \frac{2}{k-1}\Bigl(\frac{k-2}{k}\Bigr)^2E[(\mu_1 - \mu)^2] \\
         && + \frac{k-2}{k-1}\Bigl(\frac{2}{k}\Bigr)^2E[(\mu_1 - \mu)^2] \\
         &=& \frac{1}{k}\sigma_i^2 + \frac{k-2}{2}\sigma_i^2 + \frac{2}{k-1}\left(\frac{k-2}{k}\right)^2\frac{k}{2}\sigma_i^2 \\
         && + \frac{k-2}{k-1}\left(\frac{2}{k}\right)^2\frac{k}{2}\sigma_i^2 \\
     &=& \left(\frac{1}{k} + \frac{k-2}{2} + \frac{(k-2)^2}{k(k-1)} + \frac{2(k-2)}{k(k-1)}\right)\sigma_i^2 \\
     &=& \frac{2k - 2 + k(k-2)(k-1) + 2(k-2)^2 + 4k - 4}{2k(k-1)}\sigma_i^2 \\
     &=& \frac{k^3 - k^2 - 2}{2k(k-1)}\sigma_i^2.
\end{eqnarray*}
Note that the the variance component must be computed with $k-1$ and not $k$ because of the finite nature of $k$ in this case.

Putting it all together,
\begin{eqnarray*}
Var(B_i) &=& \frac{1}{k}Var(B_i|j \to j) + \frac{k-1}{k} Var(B_i | j \to j') \\
         &=& \frac{1}{k}\left(\frac{k}{2} + \frac{1}{k}\right)\sigma_i^2 + \frac{k-1}{k}\frac{k^3 - k^2 - 2}{2k(k-1)}\sigma_i^2 \\
        &=& \left(\frac{1}{2} + \frac{1}{k^2} + \frac{k^3 - k^2 - 2}{2k^2} \right)\sigma_i^2 \\
        &=& \frac{k}{2}\sigma_i^2.
\end{eqnarray*}
\end{proof}

DDML, due to its ``Draw and Discard'' update strategy, dissipates the additional intra-model variability through random model discarding which is particularly important when the model has converged and the contractive pull of the SGD is either small or non-existent, at a time when we continue training the model and adding the Laplace noise on the client. 

\subsubsection{Asynchronous Learning.} Maintaining $k$ model instances allows for a \emph{scalable}, \emph{simple} and \emph{asynchronous} model, updating with thousands of update requests per second. It is trivial to implement, relies on its inherent randomness for load distribution, and requires no locking mechanism that would pause the learning process to aggregate or summarize results at any given time.

\subsubsection{Differential Privacy.} Due to random sampling of model instances, the DDML server architecture uniquely contributes to differential privacy guarantees as will be discussed in Section~\ref{sec:privacy}. Specifically, by keeping \emph{only} the last $k$ models from clients, discarding models at random, and avoiding server-side data batching, the DDML fulfills the goal of keeping as little data as possible on the server. Through a nuanced modeling of possible adversaries (Section~\ref{sec:adversary-modeling}) and corresponding privacy analyses, DDML is able not only to provide privacy guarantees in the local model, but also improve these privacy guarantees against weaker but realistic adversaries. 

\subsubsection{Ability to Prevent Spam without Sacrificing Privacy.} The $k$ instances are instrumental in spam and abuse prevention, which is a real and ubiquitous pain point in all major client-server deployments. Nothing prevents a client from sending an arbitrary model back to the server. We could keep track of which original instance was sent to each client; however, this would negate the server-side privacy benefits and pose implementation challenges due to asynchronicity. In DDML, having $k$ replicates of each weight $b_i$ allows us to compute their means $\mu$ and standard deviations $\sigma$ and assess whether the updated model is consistent with these weight distributions (testing whether the updated value is within $[\mu - t\sigma, \mu + t\sigma]$), removing the need to make trade-offs between privacy and anti-abuse.

\subsubsection{Improved Performance.} Lastly, averaging $k$ models for prediction naturally extends DDML into \emph{ensemble} models, which have been shown to perform well in practice. Currently, the best-performing models on the MNIST hand-written digit classification are neural net committees\footnote{http://yann.lecun.com/exdb/mnist/}. In addition, as demonstrated by our real-world example in Figure~\ref{exp1}, reasonable, practical $k\in[20, 30]$ range outperforms server-side batching of 1,000+ gradients. In fact, empirical results show (see Figure~\ref{exp1}) that equivalent performance in terms of loss and accuracy is gained when $k$ is equal to the square root of the server batching size (this is due to having made an approximately the same number of update steps after taking model discarding into account). The intuition for this is quite obvious: the convex nature of GLMs is more suitable for \emph{many small} step updates in \emph{approximate} directions than a much \emph{smaller} number of steps in \emph{precise} directions \cite{masters2018revisiting}.

\subsection{Parameter Tuning and Clipping}
Choosing the right learning rate $\gamma$ is critical for model convergence. If chosen too small, the learning process proceeds too slowly, while selecting a rate too large can lead to oscillating jumps around the true minimum. We recommend trying several values in parallel and evaluating model performance to select the best one. In the future, we plan to explore adaptive learning rate methods in which we systematically decrease $\gamma$ (and, therefore, add noise) as the model converges.

By clipping gradients to a $[\text{-}1, 1]$ range, we ensure that the sensitivity of our update is $2\gamma$. In practice, the vast majority of gradients, especially as the model becomes sufficiently accurate, are much smaller in absolute terms and, thus, could be clipped more aggressively. Clipping to a $[\text{-}0.1, 0.1]$ range would reduce sensitivity by a factor of 10 to $\gamma/5$.

\section{Privacy of DDML}\label{sec:privacy}
%\section{\TODO{Theoretical Guarantees of Privacy and Robustness}}
We now discuss differential privacy guarantees provided by DDML. Our analyses are with respect to feature-level differential privacy, as discussed in Section~\ref{sec:discussion}, but they can be easily extended to model-level privacy by scaling up the noise by the number of features or by adjusting the norm of the gradient in Algorithm~\ref{alg1}.

\subsection{Adversary Modeling}\label{sec:adversary-modeling}
The main innovation of our work with respect to privacy analyses comes from more nuanced modeling of heterogeneous adversaries, and the demonstration that the privacy guarantees a client obtains against the strongest possible adversary operating in the local model of privacy are strengthened by DDML against weaker but realistic adversaries.

Our work introduces and considers three kinds of adversaries, differing in the power of their capabilities:

%\begin{enumerate}[label=\Roman*,noitemsep,topsep=0pt]
\textbf{I (Channel Listener):} is able to observe the communication channel between the client and the server in real time and, therefore, sees both the model instance sent to the client and the updated model instance sent from the client to the server.

\textbf{II (Internal Threat):} is able to observe the models on the server at a given point in time; i.e., this adversary can see model instances $1$ through $k$ but lacks the knowledge of which of the $k$ instances was the pre-image for the latest model update due to lack of visibility into the communication channel.

\textbf{III (Opportunistic Threat):} can observe a model instance at a random time, but has no knowledge of what was the state of the model weights over the last $Tk$ updates, i.e., this adversary can, for example, see models at regular time intervals $Tk$ which is much larger than 1. Clients themselves are such threats as they periodically receive a model to update.

The first adversary is the most powerful, and the privacy guarantees we provide against this adversary (Section~\ref{sec:privacy-local}) correspond to the local model of privacy commonly considered by the differential privacy community (Section~\ref{sec:related}).

The second adversary is modeling ability to access the $k$ model instances from within the entity running DDML. It is reasonable to assume that such an adversary may be able to obtain several snapshots of the models, though it will be technically infeasible and/or prohibitively costly to obtain snapshots at the granularity that allows observation of $k$ models before and after every single update. In other words, this adversary may see the models knowing that they have just been updated with a particular client's data, but the adversary would not have any knowledge of which models was the source or pre-image for the latest update. The privacy guarantee against this adversary (Section~\ref{sec:privacy-subpoena}) will be stronger than against the Channel Listener.

The third type of adversary is the weakest and also the most common. Occasional access to models allows attackers to obtain a snapshot of $k$ model instances (in a case of an internal threat) or just a single model instance (in a case of a client who receives a model for an update) after a reasonably large number of updates $Tk$. Because $Tk$ independent noise additions have occurred in the meantime, each model instance has received an expected $T$ updates and therefore, $T$ independent noise additions after a particular user's update. Every user benefits from this additional noise to a different degree, depending on the the order in which their data was ingested, and, in expectation and with high probability, enjoys significantly stronger differential privacy guarantees against this adversary than those of the local model, as will be shown in Section~\ref{sec:privacy-opportunistic}.

The table below summarizes the privacy results of this section:
\begin{tabular}{ c | c | c | c}
  Adversary & I & II & III \\
  \hline
  Expected Privacy Guarantee & $\epsilon$ & $\frac{k-1}{2k} \epsilon$ & $\frac{1}{k \sqrt{T}} \epsilon$
\end{tabular}

\subsection{Privacy against Channel Listener (Adversary~I)}\label{sec:privacy-local}
DDML guarantees $\epsilon$-differential privacy against adversary I. The claim follows directly~\cite{dwork2014algorithmic} from our choice of the scale of Laplace noise in the client-side Algorithm~\ref{alg1} and the observation that clipping the gradient in Algorithm~\ref{alg1} ensures sensitivity of at most $2\gamma$. 

It is possible to replace the Laplace noise used in the client-side Algorithm~\ref{alg1} with Gaussian. In that case, the variance of the Gaussian noise would need to be calibrated according to Lemma 1 from~\cite{jl} or Theorem A.1 of~\cite{dwork2014algorithmic}.

%Next we discuss how the privacy guarantees improve when considering the weaker but realistic adversaries, due to its draw and discard features.

%We consider the pure differential privacy definition, which is the strongest of all the variants of differential privacy discussed in the literature and in practice.

%\subsection{Local Privacy against Adversary 1}\label{sec:privacy-local}
%Follows immediately from standard differential privacy and our choice of $\gamma$ and scale of Laplace noise in the DDML algorithm (use sensitivity discussion from earlier).

\subsection{Privacy against Internal Threat (Adversary~II)}\label{sec:privacy-subpoena}
We now show that the DDML server-side design has privacy amplification properties whenever one considers an adversary of type II. We illustrate that by computing the expected privacy loss against adversary II, where the expectation is taken over the random coin tosses of the DDML server-side Algorithm~\ref{alg2} that chooses the model instances to serve and replace.

\begin{lemma}\label{lem:symmetry-corrected}
The expected privacy loss against adversary II is $\frac{k-1}{k} \cdot \frac{\epsilon}{2}$, where the expectation is taken over the random coin tosses of the DDML server-side Algorithm~\ref{alg2} that chooses the model instances to serve and replace.
\end{lemma}
\begin{proof}

Recall the DDML algorithm and the adversary model. Adversary II knows that either 1) the client's update overwrote the previous model, so the model instance prior to the update is no longer among the $k$ they can see, or 2) the client updated an existing model that is still observable among the $k$, but the adversary doesn't know which one was updated. We will call the model that was sent to the client the \textit{pre-image} and the resulting returned model $B^*$. %, and the set of $k$ models \textbf{B}.

Because of the design of DDML, the first scenario occurs with probability $\frac{1}{k}$ and the second scenario occurs with probability $\frac{k-1}{k}$. Moreover, if we are in the first scenario (i.e., the pre-image is no longer among the $k$ models due to the ``discard" part of DDML), then the client has perfect privacy against adversary II. Indeed, due to the nature of the update step in GLM, $B^*$ provides equal support for any client input when the pre-image is unknown. In other words, the privacy loss in the first scenario, $\epsilon_1$, is $0$.

We now compute the privacy loss in the second scenario, when the client updated an existing model that's still observable among the $k$, but the adversary doesn't know which one was updated. We first do the analysis for $k=2$, and then generalize it to any $k$.

In this case, the privacy loss $\epsilon_2$ is defined as \[\exp(\epsilon_2) = \max_{B^1, B^2, a, a'}  \frac{Pr[B^1, B^2 | a]}{Pr[B^1, B^2 | a']} , \] where the probability is taken over the random choices of Algorithms~\ref{alg1} and~\ref{alg2}, $B^1, B^2$ are all possible outputs in the range of Algorithm~\ref{alg1}, and $a$ and $a'$ are the private values of the client (in DDML's case, the clipped average gradients in $[-1, 1]$).

Expanding to account for the uncertainty of adversary II of whether $B^2$ is the updated model and $B^1$ -- its pre-image, or vice versa, we have
$$
{\tiny e^{\epsilon_2} = \max_{B^1, B^2, a, a'} \frac{0.5 Pr[B^1 | a, B^2 \text{ pre-image}] + 0.5 Pr[B^2 | a, B^1 \text{ pre-image}]}{0.5 Pr[B^1 | a', B^2 \text{ pre-image}] + 0.5 Pr[B^2 | a', B^1 \text{ pre-image}]}},
$$
with probabilities now being taken over the random choices of the client-side Algorithm~\ref{alg1}.

Plugging in for the noise introduced by Algorithm~\ref{alg1}, we have
\tiny
\begin{eqnarray*}
& & e^{\epsilon_2} \\
      &=& \max_{B^1, B^2, a, a'} \frac{Pr[B^1 = B^2 - \gamma a + L(0, \frac{2\gamma}{\epsilon})] + Pr[B^2 = B^1 - \gamma a + L(0, \frac{2\gamma}{\epsilon})]}{Pr[B^1 = B^2 - \gamma a' + L(0, \frac{2\gamma}{\epsilon})] + Pr[B^2 = B^1 - \gamma a' + L(0, \frac{2\gamma}{\epsilon})]} \\
      &=& \max_{B^1, B^2, a, a'} \frac{Pr[L(0, \frac{2\gamma}{\epsilon}) = B^1 - B^2 + \gamma a] + Pr[L(0, \frac{2\gamma}{\epsilon}) = B^2 - B^1 + \gamma a]}{Pr[L(0, \frac{2\gamma}{\epsilon}) = B^1 - B^2 + \gamma a'] + Pr[L(0, \frac{2\gamma}{\epsilon}) = B^2 - B^1 + \gamma a']}.
\end{eqnarray*}
\normalsize

By properties of Laplace noise,\\
$$
e^{\epsilon_2} = \max_{B^1, B^2, a, a'} \frac{\exp\Bigl(- \frac{|B^1 - B^2 + \gamma a|}{\frac{2\gamma}{\epsilon}}\Bigr) + \exp\Bigl(- \frac{|B^2 - B^1 + \gamma a|}{\frac{2\gamma}{\epsilon}}\Bigr)}{\exp\Bigl(- \frac{|B^1 - B^2 + \gamma a'|}{\frac{2\gamma}{\epsilon}}\Bigr) + \exp\Bigl(- \frac{|B^2 - B^1 + \gamma a'|}{\frac{2\gamma}{\epsilon}}\Bigr)}.
$$

Case analysis shows that the maximum is achieved when $B^1=B^2$ and $a'=-1, a=0$ or $a'=1, a=0$.
Thus, 
$e^{\epsilon_2} = \frac{2}{2 \exp(-\frac{\epsilon}{2})}$ or $\epsilon_2 = 0.5\epsilon.$

Therefore, the expected privacy loss for $k=2$ is
$$
0.5 \epsilon_1 + 0.5 \epsilon_2 = 0.25 \epsilon.
$$

Given the result for $k=2$, it can be shown that the maximum for $\epsilon_2$ in the case of $k>2$ models is also achieved when all of the model instances are equal, and the updates differ by 1 in absolute value.  Indeed, 
\tiny
\begin{eqnarray*}
&  & \exp \left(\epsilon_2\right) \\ 
 &=& \max_{B^i, B^j, a, a', {i \neq j}} \frac{\sum_{i, j, i \neq j} \left( \exp\Bigl(- \frac{|B^i - B^j + \gamma a|}{\frac{2\gamma}{\epsilon}}\Bigr) + \exp\Bigl(- \frac{|B^j - B^i + \gamma a|}{\frac{2\gamma}{\epsilon}}\Bigr) \right)}{\sum_{i, j, i \neq j} \left(\exp\Bigl(- \frac{|B^i - B^j + \gamma a'|}{\frac{2\gamma}{\epsilon}}\Bigr) + \exp\Bigl(- \frac{|B^j - B^i + \gamma a'|}{\frac{2\gamma}{\epsilon}}\Bigr)\right)} \\
 & \leq &
 \max_{B^i, B^j, a, a', {i \neq j}} \max_{i, j} \frac{\exp\Bigl(- \frac{|B^i - B^j + \gamma a|}{\frac{2\gamma}{\epsilon}}\Bigr) + \exp\Bigl(- \frac{|B^j - B^i + \gamma a|}{\frac{2\gamma}{\epsilon}}\Bigr)}{\exp\Bigl(- \frac{|B^i - B^j + \gamma a'|}{\frac{2\gamma}{\epsilon}}\Bigr) + \exp\Bigl(- \frac{|B^j - B^i + \gamma a'|}{\frac{2\gamma}{\epsilon}}\Bigr)} \\ 
& \leq & \exp \left(\frac{\epsilon}{2} \right)
\end{eqnarray*}
\normalsize

Hence, in the case of $k$ models the overall expected privacy loss is $\frac{1}{k} \epsilon_1 + \frac{k-1}{k} \epsilon_2 = \frac{k-1}{k} \cdot \frac{\epsilon}{2}$, as desired.
\end{proof}

Note that the largest privacy loss is achieved when all model instances are identical, which is consistent with intuition: when all model instances are identical, the privacy amplification against adversary II comes only from the ``discard" step of DDML; whereas when the model instances held by the server are non-identical, there's additional benefit from the uncertainty introduced by the server-side model management of Algorithm~\ref{alg2} against adversary II. Specifically, the adversary does not know which model instance was the one returned and which was the pre-image, providing additional privacy amplification. 

At first, the privacy amplification of $\frac{2k}{k-1}$ for adversary II over adversary I, may seem insignificant. This is not the case for two reasons: first, the constants in privacy loss matter, since, by the definition of differential privacy, the privacy risk incurred by an individual choosing to contribute their data grows exponentially with the growth of the privacy loss variable~\cite{nissim2017differential, Registry}. Since differential privacy started to gain traction in industry, there has been significant debate devoted to establishing what is a reasonable privacy loss rate and to optimizing the analyses and algorithms to reduce the privacy loss needed~\cite{wired2017, Tang, pate2}.

Second, the privacy loss of Lemma~\ref{lem:symmetry-corrected} is very unlikely to be realized in practice, as the scenario of all of the model instances being identical is unlikely. Specifically, the probability of how unlikely it is can be studied using the variance stabilization Theorem~\ref{thm:vars}. The argument would take the form of: with high probability due to variance stabilization, there are several model instances that are not identical and therefore, can be both the pre-image candidates or the instances returned by the client, interchangeably. The higher the number of plausible pre-image candidates among the model instances, the less certainty the adversary has about the update, and therefore, the smaller the privacy loss. We plan to formalize this intuition in future work.

\subsection{Privacy against Opportunistic Threat (Adversary~III)}\label{sec:privacy-opportunistic}
Finally, we analyze the privacy guarantees DDML provides against adversary III -- the one that is able to inspect a random model instance out of the $k$ models after a user of interest to the adversary has submitted their model instance and an expected $T$ updates to that model instance have occurred since. Note that in practice, the adversary may have an estimate of $T$, but not know it precisely, as it is difficult to measure how many updates have occurred to a model instance in a distributed system serving millions of clients such as DDML.

The privacy amplification against this adversary will come from two sources -- from the ``discard" step, in that it contributes to the possibility that the model the user contributed to is discarded in the long-term and from the accumulation of the noise, in that with each model update, additional random noise is added to it, which contributes to further privacy protection for the user whose update has occurred in preceding steps. The analysis of the privacy amplification due to the ``discard" step is presented in Lemma~\ref{lem:discard}; the analysis due to noise accumulation -- in Lemma~\ref{lem:acc-l}. 

Specifically, we find that as $T$ becomes sufficiently large, with probability $\approx 1-\frac{1}{k}$, each particular contribution has perfect privacy, and with probability $\frac{1}{k}$ the privacy guarantee is amplified by $\sqrt{T}$, resulting in the overall $k \sqrt{T}$ expected privacy amplification against adversary III compared to adversary I.

\begin{lemma}\label{lem:discard}
DDML discards $1-\frac{1}{k}$ updates long-term.
\end{lemma}

\begin{proof}
Consider a Markov process on the states $0, 1, \ldots, k$, where each state represents the number of models in which a particular update can be found. Denote by $p_i$ - the probability to go from $i$ to $i-1$ or $i+1$. In DDML, $p_i = \frac{(k-i)i}{k^2}$, for $1 \leq i \leq k-1$, and $p_0 = 0, p_k = 1$.

Let $q_i$ be the probability of eventually ending up at state 0 if you start in state $i$. By the set-up of DDML, for $1 < i<k-1$ we have:\\
$q_i = p_iq_{i-1} + p_iq_{i+1}+(1-2p_i)q_,$ or $2q_i = q_{i-1}+q_{i+1}$.

We also know that 
$q_0=1, p_0=0$, $q_k=0, p_k=0$, \\
$q_1 = p_1q_0+p_1q_2+(1-2p_1)q_1$,\\
$q_{k-1} = p_{k-1}q_{k-2} + p_{k-1}q_k+(1-2p_{k-1})q_{k-1}$ 

Summing equations for $1 \leq i \leq k-1$ we have\\
$2q_1 + \cdots + 2 q_{k-1} = q_0 + q_1 + 2(q_2+\cdots+q_{k-2})+q_{k-1}+q_{k}$
Simplifying:
$q_1+q_{k-1} = q_0+q_k$ or 
\begin{equation}\label{eq:1}
q_1+q_{k-1} = 1
\end{equation}

On the other hand, \\
$q_{k-2} = 2q_{k-1}$,\\
$q_{k-3} = 2q_{k-2} - q_{k-1} = 3q_{k-1}$ \\
$q_{k-4}=2q_{k-3}-q_{k-2} = 4q_{k-1}$\\
$q_{k-5}=2q_{k-4}-q_{k-3}=5q_{k-1}$\\ 
$q_{k-6}=2q_{k-5}-q_{k-4}=6q_{k-1}$\\ 
$\cdots$\\
\begin{equation}\label{eq:2}
q_1 = 2q_2-q_3 = (k-1)q_{k-1}
\end{equation}

Combining~(\ref{eq:1}) and~(\ref{eq:2}), we have $(k-1)q_{k-1}+q_{k-1} = 1$ or $q_{k-1}=\frac{1}{k}$ and $q_1 = 1-\frac{1}{k}$.
Since DDML is set-up that each particular contribution is initially in state 1, this completes the proof.
\end{proof}

%\begin{lemma}\label{lem:discard-old}
%After $Tk$ updates, the probability that a contribution of a particular individual is no longer present in any of the models is at least $\frac{1-z^{Tk}}{2}$, where $z=\frac{(k-1)^2+1}{k^2}$. 
%\end{lemma}
%
%%%%%%%%%%%%%%%%%%%%%%%%%%
%
%\begin{proof}[Proof of Lemma~\ref{lem:discard}]
%Consider a particular model $B$ that initially appears once among the $k$ models. At each update of Algorithm~\ref{alg2},\\
%-- with probability $\frac{k-1}{k} \cdot \frac{1}{k}$, this model gets over-written by another model; \\
%-- with probability $\frac{1}{k} \cdot \frac{1}{k} + \frac{k-1}{k} \cdot \frac{k-1}{k} = \frac{(k-1)^2+1}{k^2} = z$, there continues to be only one copy of a model derived from this model;\\
%-- with probability $\frac{1}{k} \cdot \frac{k-1}{k}$, the number of copies of models derived from this model increases.
%
%So the probability $P$ that a contribution of a particular individual is no longer present in any of the models after $Tk$ updates is:
%\[P \geq \frac{k-1}{k^2} + z \cdot (\frac{k-1}{k^2} + z \cdot (\frac{k-1}{k^2} + z (...))) = \]
%\[= \frac{k-1}{k^2} \cdot (1 + z + z^2 + \ldots z^{Tk-1}) = \frac{k-1}{k^2} \cdot \frac{1-z^{Tk}}{1-z} = \frac{1-z^{Tk}}{2}. \]
%\end{proof}
%
%In particular, as $Tk \rightarrow \infty$:
%$P \geq \frac{k-1}{k^2} \cdot \frac{1}{1-z} = \frac{k-1}{k^2 \cdot \frac{k^2-(k-1)^2-1}{k^2}} = \frac{k-1}{k^2 - k^2+2k-1-1} = \frac{1}{2},$ and even for $k=20, T=10: \frac{1-z^{Tk}}{2} \approx \frac{1}{2}$.

\begin{lemma}\label{lem:acc-l}
With high probability, DDML guarantees a user $(\epsilon_T, \delta_T)$-differential privacy against adversary III, 
where\\ $\epsilon_T = \frac{\epsilon}{\sqrt{2T}} \sqrt{\ln\left(\frac{1}{2\delta_T}\right)} $ and $\delta_T$ is an arbitrary small constant (typically chosen as O(1/ number of users)). $T$ is the number of updates made to the model instance between when a user submitted his instance update and when the adversary observes the instances. The statement holds if $T$ is sufficiently large.
\end{lemma}

\begin{proof}[Proof of Lemma~\ref{lem:acc-l}]
We rely on the result from concurrent and independent work of~\cite{Birs} obtained in a different context to analyze the privacy amplification in this case. Specifically, their result states that for \emph{any} contractive noisy process, privacy amplification is no worse than that for the \emph{identity} contraction, which we analyze below.

The sum of $T$ random variables drawn independently from the Laplace distribution with mean $0$ will tend toward a normal distribution for sufficiently large $T$, by the Central Limit Theorem. In DDML's case with Laplace noise, the variance of each random variable is $\frac{8\gamma^2}{\epsilon^2}$, therefore, if we assume that the adversary observes the model instance after $T$ updates to it, the variance of the noise added will be $T \cdot \frac{8\gamma^2}{\epsilon^2}$. This corresponds to Gaussian with scale $\sigma = \frac{2\sqrt{2T}\gamma}{\epsilon}$.

Lemma 1 from~\cite{jl} states that for points in $p$-dimensional space that differ by at most $w$ in $\ell_2$, addition of noise drawn from $N^p(0, \sigma^2_T)$, where $\sigma_T \geq w \frac{\sqrt{2\left(\ln\left(\frac{1}{2\delta_T}\right)+\epsilon_T\right)}}{\epsilon_T}$ and $\delta_T < \frac{1}{2}$ ensures $(\epsilon_T, \delta_T)$ differential privacy. We use the result of Lemma 1 from~\cite{jl}, rather than the more commonly referenced result from Theorem A.1 of~\cite{dwork2014algorithmic}, because the latter result holds only for $\epsilon_T \leq 1$, which is not the privacy loss used in most practical applications.

We now ask the question: what approximate differential privacy guarantee is achieved by DDML against adversary III? To answer this, fix a desired level of $\delta_T$ and use the approximation obtained from the Central Limit theorem to solve for the $\epsilon_T$.

\[\frac{2\sqrt{2T}\gamma}{\epsilon} \geq w \frac{\sqrt{2\left(\ln\left(\frac{1}{2\delta_T}\right)+\epsilon_T\right)}}{\epsilon_T}\]
\[T \cdot \frac{8\gamma^2}{\epsilon^2} \geq w^2 \frac{2\left(\ln\left(\frac{1}{2\delta_T}\right)+\epsilon_T\right)}{\epsilon^2_T}\]

%\[T \cdot \frac{8\gamma^2}{\epsilon^2} \cdot \epsilon^2_T - 2w^2 \epsilon_T - 2w^2 \ln\left(\frac{1}{2\delta_T}\right) \geq 0\]
\[T \cdot \frac{4\gamma^2}{\epsilon^2} \cdot \epsilon^2_T - w^2 \epsilon_T - w^2 \ln\left(\frac{1}{2\delta_T}\right) \geq 0\]

Solving the quadratic inequality, we have:
\[D = w^4+4T \cdot \frac{4\gamma^2}{\epsilon^2} w^2 \ln\left(\frac{1}{2\delta_T}\right)\]

\[\epsilon_T \geq \frac{w^2 + \sqrt{D}}{2T \cdot \frac{4\gamma^2}{\epsilon^2} } = \frac{\epsilon^2 w^2}{8\gamma^2T} \sqrt{1+16T\frac{\gamma^2}{\epsilon^2 w^2} \ln\left(\frac{1}{2\delta_T}\right)}\]

For large $T$, the additive term of $1$ under the square root is negligible, so we have:
\[\epsilon_T \approx \frac{\epsilon^2 w^2}{8\gamma^2T}  4\sqrt{T} \frac{\gamma}{\epsilon w} \sqrt{\ln\left(\frac{1}{2\delta_T}\right)} =  \frac{\epsilon w}{2\gamma \sqrt{T}} \sqrt{\ln\left(\frac{1}{2\delta_T}\right)}\]

In DDML, $w=\sqrt{2}\gamma$, therefore, 

$$
\epsilon_T \approx \frac{\epsilon}{\sqrt{2T}} \sqrt{\ln\left(\frac{1}{2\delta_T}\right)}
$$
\end{proof}

Consider a choice of $\delta_T = 10^{-8}$. Then Lemma~\ref{lem:acc-l} states that $\epsilon_T \approx \frac{3\epsilon}{\sqrt{T}}$. In other words, if DDML guarantees $\epsilon$ pure differential privacy in the local model against an adversary who can observe the channel communication between the client and the server, then it provides a $(\frac{3}{\sqrt{T}} \epsilon, 10^{-8})$-approximate differential privacy guarantee against the weaker adversary who can observe the models after $T$ updates, with high probability.  Although pure and approximate differential privacy are not directly comparable, one interpretation of this result is that the privacy loss guarantee against the more realistic and more common adversary III improves inversely proportional to the square root of the frequency with which that adversary can observe the model instances or, equivalently, the speed with which the adversary can ensure that he obtains the models after he knows that the target user of interest has sent the server an instance updated with their data. Even though the inverse dependence of privacy loss against adversary III on $\sqrt{T}$ is only approximate (see~\cite{Birs} for a precise analysis), it is noteworthy, as in practical applications it may allow choosing a higher $\epsilon$, and thus improve the utility and convergence speed of the learning framework. Indeed, adversary I is extremely unlikely and/or requires significant effort to implement; therefore, it may be acceptable to have a higher privacy loss against it, while simultaneously maintaining a sensible privacy loss against the more realistic and less resource intensive adversary III.

\section{Experiments and Results}\label{sec:experiments}
%\section{Experiments and Results}
\begin{figure*}[!t]
\centering
\includegraphics[width= .9\linewidth]{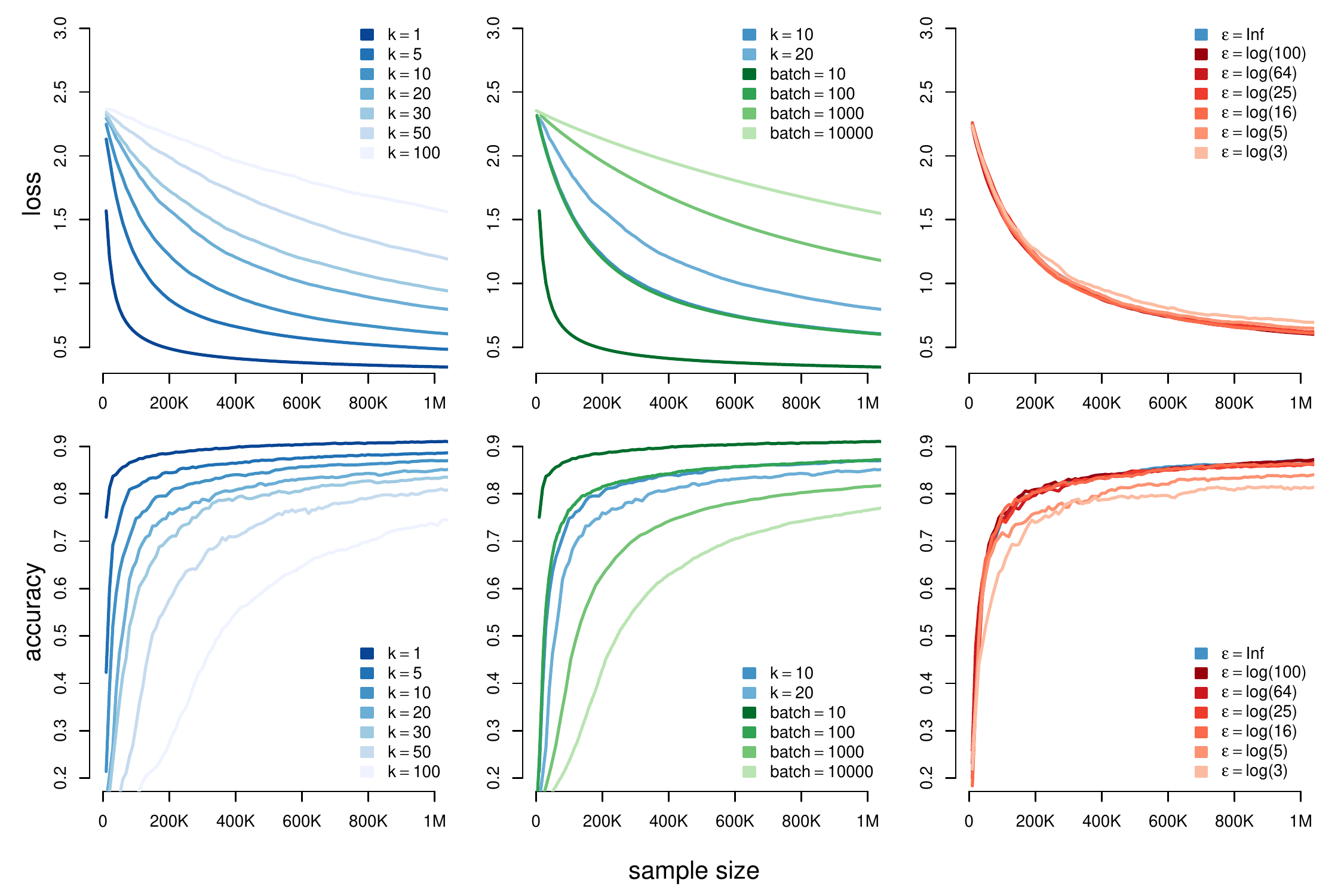}
\caption{Performance of DDML for different $k$ (1st column) and $\epsilon$ (3rd column). We also compare server-side batching and $k$ model instances (2nd column). Learning rate is $\gamma=0.001$.}
\label{exp1}
\end{figure*}

We study performance of the DDML framework using a multi-class logistic regression. We evaluate the impact of different choices of $k$ (the number of instances) and different levels of desired privacy budget, $\epsilon$, on both model loss (training set) and accuracy (hold-out set). In addition, we compare convergence properties of ``Draw and Discard'' with the server-side batching approach. By server-side batching, we mean a server model in which updates are streamed into a temporary storage (on a server), accumulated, and then averaged out to make a single model update.

We also evaluate DDML ``Draw and Discard'' strategy relative to alternative update schemes, which include sampling updates with probability $\frac{1}{k}$, always replacing the same instance and averaging models before overwriting them.

Finally, we demonstrate the effectiveness of DDML on training a small logistic regression with known weights in our production, real-world, distributed environment with millions of daily users.

\subsection{Experiment Configurations}
We conduct our study on MNIST digit recognition~\cite{lecun1998gradient} challenge and use multi-class logistic regression with Cross-Entropy loss as our model across different DDML configurations. The MNIST dataset has a training set of 60,000 examples, and a test set of 10,000 examples. Each 28x28 image is serialized to a 784 vector that serves as features for predicting 0-9 hand-written digits.

We set our learning rate $\gamma = 0.001$ for all experiments and standardize all features to $[0, 1]$. We simulate ``clients'' by randomly assigning $N_c = 10$ examples to each one, resulting in 6,000 ``clients''. We also make 20 passes over training data for all configurations. Because the number of model updates differs for different experiment configurations, we standardize our experiments by looking at the sample size, i.e., the number of data points ingested by the algorithm up to that point. Results of our experiments are visualized in Figure~\ref{exp1}. For experiments where $k > 1$, we initialize them using normal distribution with means $0$ and variance $\frac{k}{2}\sigma^2$, where $\sigma^2$ is the client-side noise variance. For cases when we do not add noise on the client side, we initialize with $\epsilon=1$ in the client-side noise calculation. 

\textbf{Comparing $k$s: } Because the number of model instances $k$ is so central to the DDML framework, its impact on model performance must be carefully evaluated. When studying the effect of $k$, we did not add any Laplace noise on the client side to eliminate additional sources of entropy. Loss on a training set for different $k$s from 0 to 100 are shown in the first panel of Figure~\ref{exp1}. Accuracy results on the test set using the averaged model over the $k$ instances is shown in the first panel of the second row. The $k=1$ configuration is equivalent to a standard server-side model training (the darkest blue line) and clearly has the best performance. As we add more model instances, the number of updates to each model instance decreases (it's equal to $M/k$ in expectation), and averaging over $k$ instances, though beneficial, is not sufficient to make up for the difference. Of course, as $M$, the sample size, goes to infinity, all configurations converge and have equivalent performance metrics\footnote{We do not demonstrate this in practice. It follows from theoretical optimization results on convex functions.}. In practice, we have been using $k \in [20, 30]$ model instances, which is sufficient for scalability and does not incur a large hit in terms of short- to medium-term performance.

\textbf{Server-side batching:} A commonly used solution to server scalability problems is to batch 1,000s or 10,000s of gradients returned by clients on the server. In addition to being inferior in terms of privacy because the batch size is usually orders of magnitude larger than our $k$, we empirically demonstrate that this approach is also inferior in terms of finite-sample performance whenever server-side batch size $N_s > k^2$. The second column in Figure~\ref{exp1} compares $k=10$ and $k=20$ with four batch sizes, $N_s$, of 10, 100, 1,000, 10,000. It's remarkable how overlapping $k=10$ and $N_s = 100$ curves are (it holds empirically for other $k = \sqrt{N_s}$ combinations). Because the learning process must pause to average out gradients, $N_s$ is usually much higher than $k$ to accommodate the high-throughput traffic, which can be easily handled with just $k=20$ model instances and no interruptions to the learning process. There is an argument to be made here that one would use a larger learning rate for the server-side batching strategy, since the gradients are estimated more accurately and larger updates would be warranted. However, it is outside of the scope of this work to study the feasibility of this claim in very general terms. If the learning rate $\gamma$ is selected ``optimally" for the DDML strategy, it is hard to imagine that one could \emph{always} make it much larger without getting into oscillation problems.

\textbf{Privacy parameter $\epsilon$: } For this set of experiments, we fix $k=10$ and vary the amount of noise added on the client. Results are shown in the last column of Figure \ref{exp1}, with the blue line indicating model performance without the noise (same as $k=10$ in the first panel). We observe that for $\epsilon > \log(16)$ there is no substantial negative impact of providing client-side privacy on the model's performance, while smaller privacy parameter choice values such as $\epsilon=\log(3)$ do have some negative impact.

\subsection{Comparing to Alternative Server Architectures}
``Draw and Discard`` architecture may be criticized for being too wasteful with data in terms of permanently losing too many updates due to its overwriting strategy. As is clear from Lemma~\ref{lem:discard}, DDML keeps \emph{only} $\frac{1}{k}$ model updates; however, in exchange, as we will explain below, it receives scalability and generous privacy benefits, which is ultimately the main reason of going down the path of training in a distributed manner. It is, however, critical to understand the reasons behind the short- to medium-term drop in model performance relatively to server-side training. As we demonstrate in the experiments below, it does not come from ``losing`' data, but from partitioning the training data into $k$ groups, ultimately resulting in $\frac{M}{k}$ updates to each of them. 

We compared the DDML server-side architecture (``Draw and Discard'') to several alternatives. Here, we use $k=20$, $\gamma = 0.001$ and no Laplace noise to make comparisons easier to interpret and comparable to Figure \ref{exp1}. First, in terms of data utilization rate, DDML with random replacement of instances should be approximately equivalent to accepting 1 in $k$ model updates (``$\frac{1}{k}$ Update Accept Rate''), i.e., if $k = 20$, we would reject 19 model updates and accept 1, on average. We confirm in Figure~\ref{fig:strategies} that, indeed, this is the case (green and pink lines essentially overlap and have very similar accuracy rates on the MNIST dataset). 

It is very interesting to note that the strategy of sampling a model instance at random and always overwriting exactly the same one after an update (``Same Instance Replace'') is essentially equivalent to both the ``Draw and Discard'' and ``$\frac{1}{k}$ Update Accept Rate'' strategies as can be seen in Figure \ref{fig:strategies}. In this case, we do not discard any model updates, yet the short-term performance does not improve. The reason for that is that all $k$ model instances received $\frac{M}{k}$ updates in expectation, and ultimately arrived to similar locations in terms of objective function optimization. Averaging helps to stabilize the estimates, but cannot possibly move the system closer to the global minimum.

We also tried the averaging strategy where before overwriting the model, we averaged weights between the updated model and the model that is about to be overwritten. In a completely sequential update case, this results in \emph{decreased} performance for both ``Draw and Discard" and ``Same Instance Replace" as can be seen in Figure \ref{fig:strategies}. This surprising result can be easily explained by the fact that by taking an average, we effectively decrease the learning rate (in the case of ``Same Instance Replace" strategy, we make it exactly half) which results in inferior performance if the learning rate was not too large in the first place.

Given that the current ``Draw and Discard'' strategy performs as well or better as any other strategy considered, the relevant considerations when deciding which to choose are privacy, ease of implementation, latency, and robustness to spam and abuse. We observe that, against attacker II, DDML offers a factor of $2 \frac{k}{k-1}$ privacy amplification to all users, whereas ``$\frac{1}{k}$ Update Accept Rate" offers perfect privacy to $(1-\frac{1}{k})$-fraction of updates and no privacy amplification for $\frac{1}{k}$ fraction of updates; thereby, arguably making DDML superior. DDML is trivial to implement in a distributed cloud architecture such as Google Cloud Platform. A one-model approach is much harder to implement and maintain, as one has to deal with locking, situations when the client does not return a model, etc. Having multiple versions of the same model allows for private spam and abuse detection without the need to know the pre-image.

% while providing increased privacy due to having $k$ model instances and is easy to implement and maintain, it is the one we use in practice. 

\begin{figure}[!t]
\centering
\includegraphics[width= .9\linewidth]{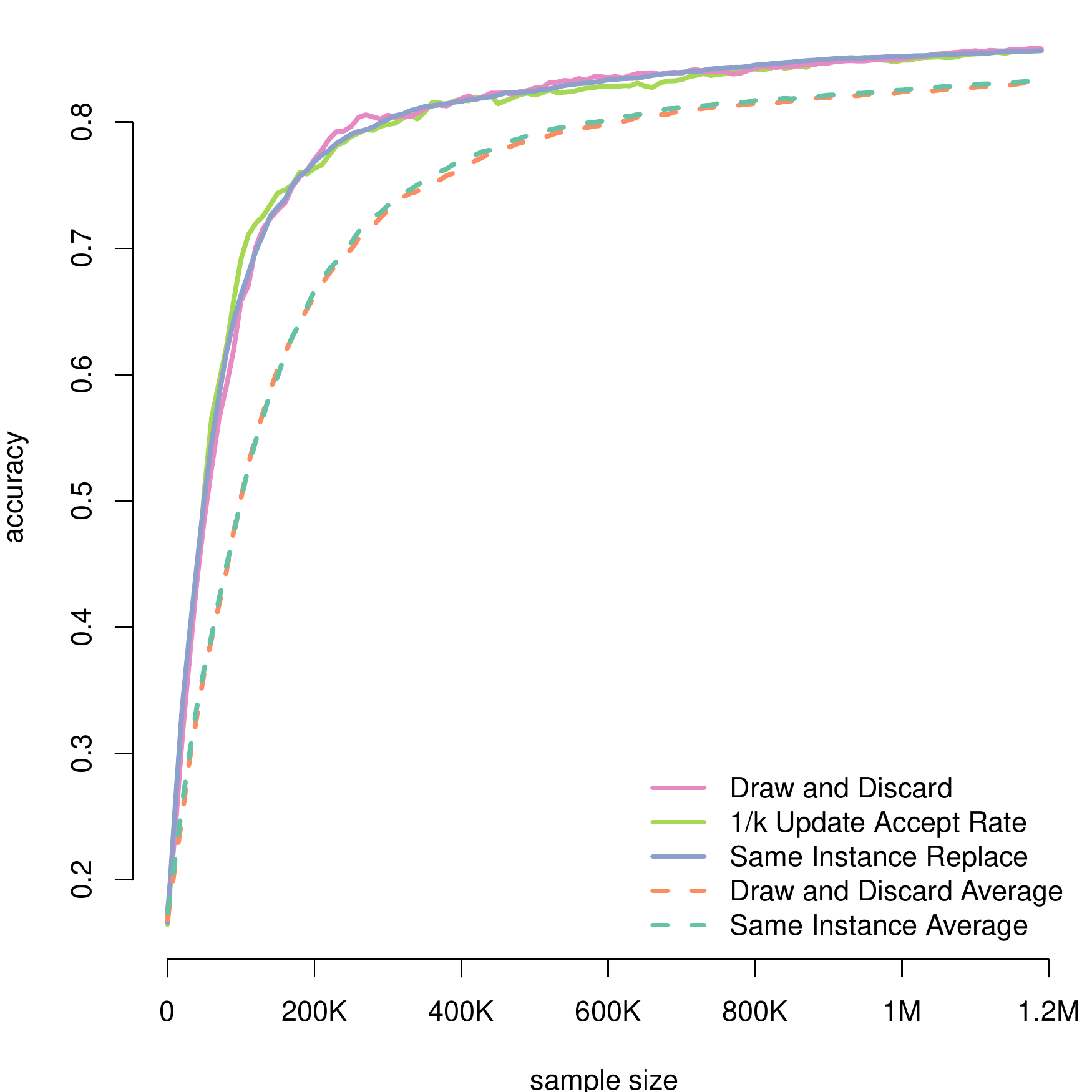}
\caption{The proposed DDML update strategy ``Draw and Discard" performs as well or better in terms of accuracy than several alternative update strategies we considered. In this case, we use $k = 20$ and learning rate $\gamma = 0.001$. In particular, it is interesting to note that it is essentially equivalent to replacing the same instance after an update and superior to averaging strategies.}
\label{fig:strategies}
\end{figure}

\begin{figure*}[!ht]
\centering
\includegraphics[width= .9\linewidth]{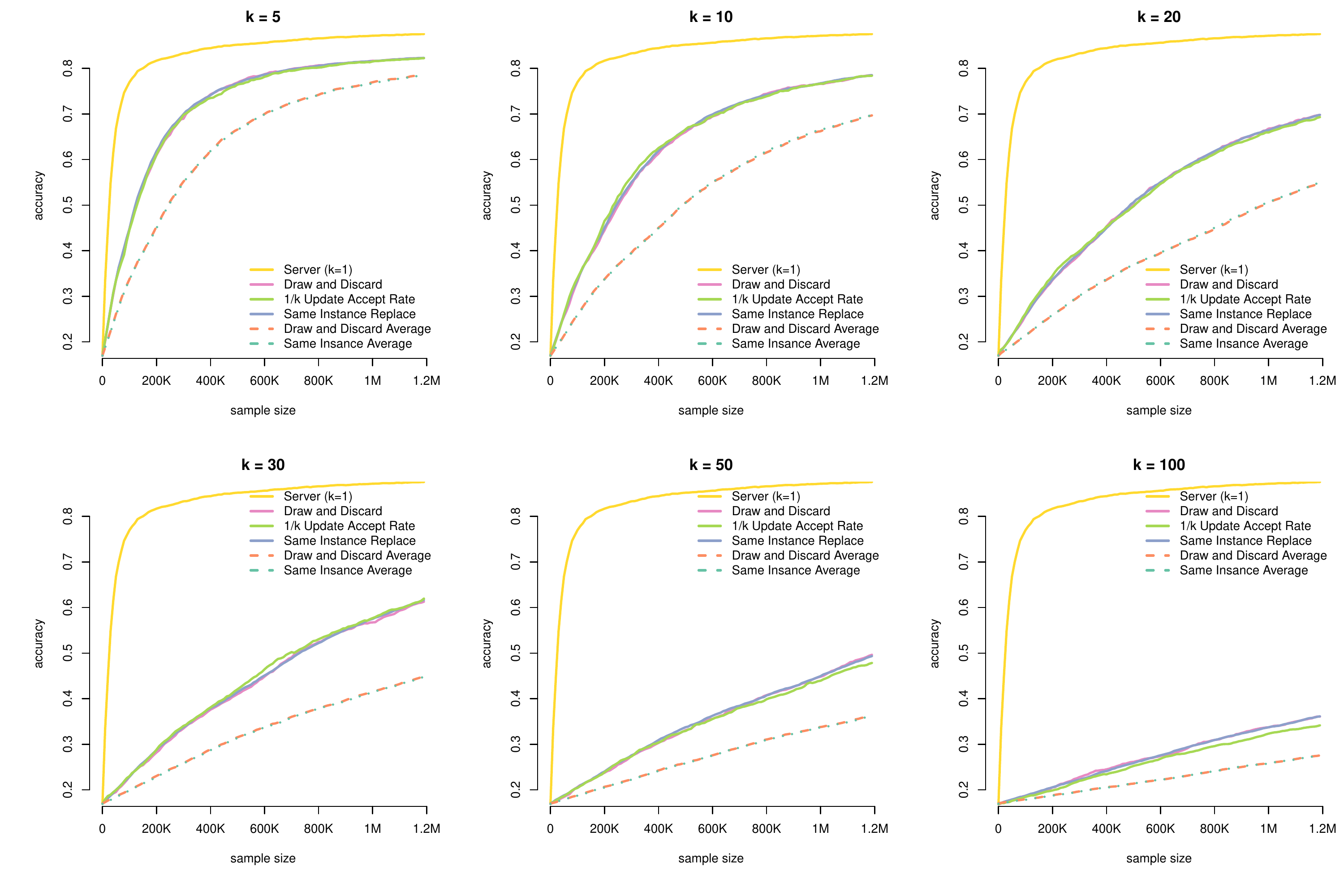}
\caption{We compare several update strategies, varying the number of model instances $k$. Again, we observe results analogous to Figure \ref{fig:strategies}. Here, learning rate $\gamma = 0.0001$ which is 10 times smaller than the one we used in Figure \ref{exp1}.}
\label{fig:smallrate}
\end{figure*}

We ran the above strategies for a whole range of $k$'s and also decreased the learning rate by a factor of 10 compared to simulations in Figure \ref{exp1}. Results are shown in Figure \ref{fig:smallrate}. We would like to make two observations. First, relative performance of different strategies remains similar regardless of $k$. Second, decreasing the learning rate results in substantially worse performance for the ``Draw and Discard'' strategy relative to the server-side one, especially as $k$ becomes larger. In fact, we are observing that learning rate has a larger effect on performance than the number of model instances $k$. Given how hard it is to select an appropriate learning rate in practice, especially in a distributed system, and situations when the choice is off by a factor of 10 from the optimal one not uncommon, any loss of efficiency due to having $k$ instances could be insignificant in comparison.

\subsection{Logistic Regression With Known Weights Trained On Distributed Devices}

We have trained a simple logistic regression with 33 weights using DDML on a fully functional DDML framework in our production environment with millions of daily users. The model consists of 5 distinct predictors $B$ (boolean), $E_1$ (categorical: 3 categories), $E_2$ (categorical: 4 categories), $D$ (double) and $I$ (integer) and all possible two-way crosses between them. The exact model specification is shown below. 
\small
\begin{eqnarray*}
\mathtt{logit}(P(O)) &=& 0.34 -1.18B + 1.05E_{11} + 1.6E_{12} \\
		 &-& 1.51E_{21} + 0.72E_{22} + 1.36E_{23} \\
         &-& 1.41 D -1.67I \\
         &-& 1.14BE_{11} + 1.3BE_{12} \\
         &+& 1.72BE_{21} + 0.18BE_{22} -1.52BE_{23} \\
         &+& 1.72BD -1.17BI \\
         &-& 0.12E_{11}E_{21} -0.53E_{11}E_{22} + 1.00E_{11}E_{23}\\
         &-& 0.42E_{12}E_{21} -1.64E_{12}E_{22} -1.53E_{12}E_{23} \\
         &+& 0.64DE_{11} + 0.48DE_{12} \\
         &-& 0.17IE_{11} + 0.44IE_{12} \\
         &-& 0.51BE_{21} -0.68BE_{22} + 1.65BE_{23} \\
         &+& 0.51IE_{21} -1.64IE_{22} + 0.57IE_{23} \\
         &-& 0.80DI
\end{eqnarray*}
\normalsize

Given our selection of parameters and features, the true probability of outcome (O) for this model is about 46\%. On a client side, we store six fields in a local SQL database. The model specification, in terms of which features to use and which crosses to include, is controlled completely by the server, which allows to train different models using the same underlying data without making client-side changes. An example of several rows of data for this model is shown in the table below.

\begin{center}
\begin{tabular}{ c c c c c c }
O & B & $E_1$ & $E_2$ & D & I \\
\hline
1 & 1 & "A" & "c" & 2.18 & 3 \\
0 & 1 & "B" & "a" & 0.89 & 4 \\
1 & 0 & "C" & "b" & 0.23 & 0 \\
1 & 0 & "C" & "d" & 1.30 & 1 \\
\ldots & \ldots & \ldots & \ldots & \ldots & \ldots 
\end{tabular}
\end{center}

\begin{figure}[!ht]
\centering
\includegraphics[width= .9\linewidth]{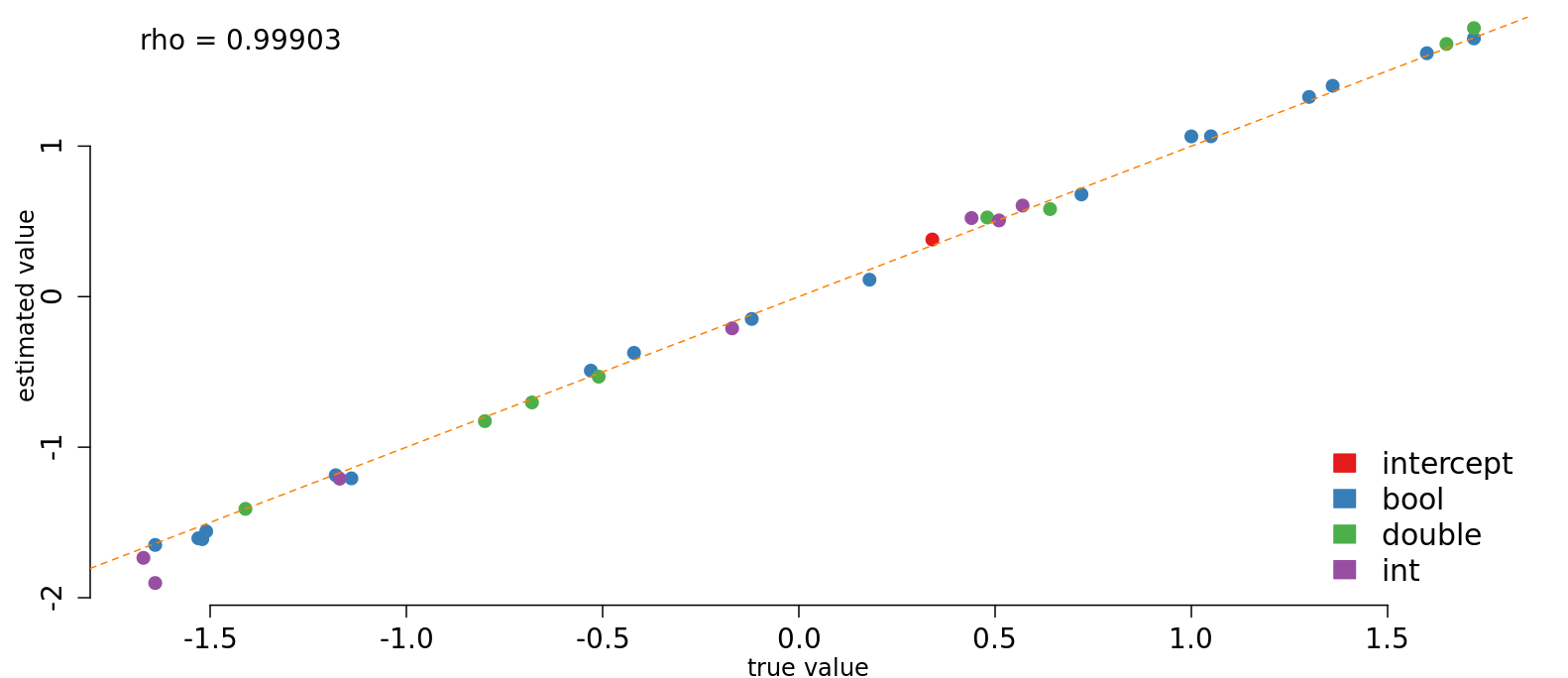}
\caption{Comparing true and estimated weights for the logistic model with known parameters (weights).}
\label{fig:rand2}
\end{figure}

DDML parameters for this model are $k=30$ and learning rate $\gamma=0.005$. We did not add noise to the update step in this case to eliminate an additional source of uncertainty, as this was the first model we trained using the framework. Up to this point, we have received 983,068,969 model updates and estimated weights along with the true ones are plotted in Figure \ref{fig:rand2}. Estimates are very close to the true parameter values. The accuracy of the model is 80\%, with precision of 79\% and recall of 72\%. 
                
We once attempted to train this model with $k=3$ model instances which resulted in hot keys in Google memcache \cite{memcache}. We received a prompt ticket from Google SREs to \emph{immediately} fix this as it was causing broader in scope (outside of our project) issues.

\section{Real World Applications}\label{sec:applications}
\subsection{Predicting Phishy URLs in Chat}
One of the most common abuse vectors for our platform is phishing user credentials. When malicious accounts manage to befriend legitimate ones or some accounts get hacked, they begin to send phishing url's to users within their social circle. Even if only a small fraction of users take the bait, the cycle of abuse is perpetuated indefinitely. We rely on Google Safe Browsing \cite{safebrowsing} and internally maintained lists of known phishy URLs to prevent these kinds of attacks. However, the process of discovering \emph{novel} malicious URLs depends on customer ops reports which is slow, manual and opportunistic.

We implemented a DDML logistic regression model for predicting the likelihood of a URL being phishy. Our features come from both the URL itself (different parts of the ULR, such as domain, subdomain, host, path, query parameters), as well as the page content. We consider features related to character distribution, special characters ($?,\&,\ldots$), lengths, language, particular keywords, etc. For content page features, we include features related to major page component, such as iframes, input boxes, password boxes, images, scripts, page size, readability and many others. 

For the model that we currently use in production, we received 1,730,624,961 model updates. The number of weights is 387. We use $k=20$,  $\epsilon = \log(32)$ and learning rate $\gamma=0.001$. 

Performance of our model is shown in the table below. 
\begin{center}
\begin{tabular}{ l l }
ROC AUC	& 0.9702 \\
PR AUC	& 0.7851 \\
Accuracy &	0.9263 \\
Precision & 0.1577 \\
Recall & 0.8922 \\
F1 & 	0.2681 
\end{tabular}
\end{center}

Overall, our recall is very good, while precision remains relatively low because of the inherent complexity of the adversarial setting in which attackers constantly change their approaches making it hard to detect new types of URLs that have not been previously seen. Our online model allowed us to increase discovery of new phishing urls by an orders of magnitude since we launched it. We are exploring the addition of features related to the way pages look, which is getting closer to image analysis and, potentially, non-linear models in the future.

\subsection{Ranking Application Based on User Text Input}
Currently, we also have an application in production ranking about 7,000 items, tagged with string qualifiers, based on user text input. The model consists of uni-grams, bi-grams, tri-grams and their selected crosses between user input and tags (~50,000 weights and 1MB in size). We have successfully deployed smaller versions of this model globally and saw an improvement in user engagement in single percentage points. The hope is that, in the future, we will obtain even higher gains with models of larger capacity.

\section{Discussion}\label{sec:discussion}
Having deployed DDML at scale at a large company with millions of daily active clients, we realized how critical a well-designed server-side architecture was to the client-side learning process. Due to the symmetric nature of draws and discards, with the number of reads equaling the number of writes, there must be sufficient redundancy in place to scale our serving infrastructure. $k$ model instances offer, besides increased privacy, an incredibly scalable and asynchronous solution to client-server communication. 

One can easily make an argument that replacing model instances at random is ``wasteful" from the model training perspective. It is partially true. However, so is setting the wrong learning rate, mismanaging the server-side batch size, etc. We are never perfect in utilizing our data in the absolute sense even before moving to a distributed ML setting. There, things only become more complicated from a learning perspective and it is not unreasonable to see additional performance sacrifices. If your architecture can support only 10 writes per second, for example, and your overall traffic is 100 write requests per second, you will not be able to perfectly utilize all your available data in a sequential updating scheme. Trade-offs must be made. The focus should be not on what we are losing because we must lose something, but what we are gaining in exchange. By making a small sacrifice in performance by introducing $k$ instances, we gain scalability, ease of implementation, spam detection, and additional privacy. The only question is whether we are trading off these properties efficiently.

DDML framework can easily be extended to support neural networks and any other models whose objective function can be written as a sum of differentiable functions. Very recent work of~\cite{masters2018revisiting} may be useful in providing guidance for parameter tuning in those cases. Extending it to decision trees seems harder and further research into distributed optimization of trees is needed.

% A DDML-like approach is appropriate for situations when data is plentiful, and ensuring local differential privacy, ease of implementation and latency, rather than efficiency in using the data, is the primary consideration.

% \subsection{The Rate of Privacy Loss}\label{sec:related-finite}
A major struggle in the application of DP to practice~\cite{wired2017, Tang} is the question of how to keep the overall per-user privacy loss~\cite{nissim2017differential, Registry} bounded over time. This is particularly challenging in the local model as more data points are needed to achieve the comparable level of utility in the local model than in the trusted curator model \cite{kairouz2014extremal}. Our approach mimics the one taken by Apple's deployment of DP~\cite{Apple, Tang}: ensure the privacy loss due to the collection of one piece of data from one user is bounded, but allow multiple collections from the same user over time. Formally, this corresponds to the privacy loss growing as the number of items submitted, as per composition theorems~\cite{kairouz2017composition, dwork2014algorithmic}. 

% \subsubsection{Feature-level Privacy}\label{sec:feature-level}
We offer \emph{feature-level} local differential privacy and, therefore, in a situation when features are correlated, the privacy loss scales with the number of features. In principle, if one would like to achieve model-level privacy, one needs to scale the noise up according to the number of features included in the model.
%This challenge is not addressed by other industry applications of differential privacy, and remains an important direction for future work. 
Applications of differential privacy to very high-dimensional data, particularly, in the local model, have not yet been adopted in practice. In theoretical work, the distinction is often mentioned, but the choice is left to industry practitioners. %This choice of whether to ensure feature-level or user-level privacy is also industry applications of differential privacy, 
We believe that in practice, feature-level privacy combined with limited server-side model retention is sufficient to protect the privacy of our clients against most realistic adversaries.

% \subsubsection{Tighter Privacy Guarantees via Other Variants of Differential Privacy and Better Adversary Modeling} 
With respect to privacy guarantees, the main focus of this work has been to ensure the strongest possible form of privacy -- pure differential privacy in the local model. We have also discussed more realistic adversary models and the way that DDML provides even better privacy guarantees against those. We are optimistic that further improvements, both in utility and in the tightness of privacy analyses are possible via switching from Laplace to Gaussian noise in the DDML client-side Algorithm~\ref{alg1}, further precision of adversary modeling, and in performing the analyses using the variant of Renyi Differential Privacy~\cite{RenyiDP} or Concentrated Differential Privacy~\cite{CDP}. The optimism for the first claim stems from experience with other differentially private applications and the very recent work of~\cite{Birs}; for the second -- from the similarities between DDML and the shuffling strategy of~\cite{Prochlo} and privacy amplification by sampling exploited by~\cite{li2012sampling, amp, Abadi2016}; for the third -- from recent work in differentially private machine learning \cite{Abadi2016, Geumlek2017, mcmahanlanguage, wu2017bolt} that benefits from analyses using such relaxations. We hypothesize that accuracy and privacy of DDML can also be improved using hybrid approaches to privacy that combine locally private data with public data or data obtained with privacy in the non-local model, as proposed by~\cite{Birs} and~\cite{Blender}.

In a sense, DDML can be viewed as a system, whose particular components, such as the approach chosen to ensure local privacy and the analysis under the chosen adversary model can be varied depending on application and the desired nuance of privacy guarantee.

\section{Conclusions}
Client-side privacy-preserving machine learning is still in its infancy and will continue to be an active and important research area both in ML and privacy for the
foreseeable future. We believe that the most important contribution of this work is a completely new server-side architecture with
random ``draw and discards'', that offers unprecedented scalability with no interrupts to the learning process, client-side and server-side
privacy guarantees and, finally, a simple, inexpensive and \emph{practical} approach to client-side machine learning. 

Our focus on simpler, yet useful in practice, linear models allowed us to experiment with client-side ML without having to worry about convergence
and other ML-related issues. Instead, we have sufficient freedom to zoom in on privacy considerations, build simple and scalable infrastructure
and leverage this technology to improve mobile features for millions of our users.

\balance
\bibliography{ddml_bib}
\bibliographystyle{IEEEtranS}
%\bibliographystyle{ACM-Reference-Format}

%\newpage
%\appendix
%\section{Appendix}
%\input{old_amplification_proofs}
%\input{appendix_code}
%\balance
%\input{appendix}
%%%%%%%%%%%%%%%%%%%%

\end{document}